\documentclass[11pt, journal]{article}
\usepackage{authblk}
\usepackage{amsmath}
\usepackage{amsthm}

\usepackage{graphicx}
\usepackage{framed, color} 
\usepackage[margin=.7in]{geometry}
\usepackage{pdflscape}
\definecolor{shadecolor}{rgb}{1, 0.8, 0.3}

\usepackage{amssymb}
\usepackage{amsfonts}
\usepackage{cite}
\usepackage{ifthen}
\usepackage{epsfig}
\usepackage{subcaption}
\usepackage{array}
\usepackage{enumerate}
\usepackage{algorithm}
\usepackage[all]{xy}
\usepackage{algpseudocode}
\usepackage[compact]{titlesec}

\usepackage{times}
\usepackage{tikz}
\usepackage{float}
\usetikzlibrary{shapes,positioning,arrows,automata}
\usepackage{caption}
\usepackage{nth}
\usepackage{array} 

\theoremstyle{plain}
\newtheorem{theorem}{Theorem}[section]
\newtheorem{lemma}[theorem]{Lemma}

\theoremstyle{definition}

\theoremstyle{remark}
\newtheorem{remark}{Remark}

\newcommand{\la}{\lambda}

\newcommand{\beq}{\begin{equation}}
\newcommand{\eeq}{\end{equation}}
\newcommand{\bea}{\begin{eqnarray}}
\newcommand{\eea}{\end{eqnarray}}
\newcommand{\bean}{\begin{eqnarray*}}
\newcommand{\eean}{\end{eqnarray*}}
\newcommand{\bit}{\begin{itemize}}
\newcommand{\eit}{\end{itemize}}
\newcommand{\ben}{\begin{enumerate}}
\newcommand{\een}{\end{enumerate}}
\newcommand{\blem}{\begin{lem}}
\newcommand{\elem}{\end{lem}}
\newcommand{\bthm}{\begin{thm}}
\newcommand{\ethm}{\end{thm}}
\newcommand{\bpf}{\begin{IEEEproof}}
\newcommand{\epf}{\end{IEEEproof}}

\newcommand{\comment}[1]{}

\newcommand\defeq{\mathrel{\overset{\makebox[0pt]{\mbox{\normalfont\tiny\sffamily def}}}{=}}}

\newcommand*\zoset{\{0,1\}}
\newcommand*\bigOh[1]{$\mathcal{O}(#1)$}
\newcommand*\ol[1]{\overline{#1}} 
\newcommand*\mbf[1]{\mathbf{#1}} 

\newcommand*\bx{\mathbf{x}} 

\newcommand*\ba{\mathbf{a}} 
\newcommand*\by{\mathbf{y}} 
\newcommand*\bz{\mathbf{z}} 
\newcommand*\bb{\mathbf{b}} 
\newcommand*\bu{\mathbf{u}}
\newcommand*\bt{\mathbf{t}}

\newcommand*\bw{\mathbf{w}}
\newcommand*\be{\mathbf{e}}

\newcommand*\set[1]{\text{set}\left({#1}\right)}
\newcommand*\supp[1]{\text{supp}\left({#1}\right)}
\newcommand*\thth[1]{$#1^{\text{\tiny th}}$}

\newcommand*\nod{K} 
\newcommand*\noi{n} 
\newcommand*\nom{m} 
\newcommand*\nob{M} 

\newcommand*\K{\nod} 
\newcommand*\n{\noi} 
\newcommand*\m{\nom} 
\newcommand*\M{\nob} 
\newcommand*\h{h} 

\newcommand*\eps{\epsilon} 

\usepackage{algorithm}\usepackage{algpseudocode}
\usepackage{mathtools}

\begin{document}

\title{\textbf{SAFFRON}: A Fast, Efficient, and Robust Framework \\
for Group Testing based on Sparse-Graph Codes}

\author{Kangwook Lee, Ramtin Pedarsani, and Kannan Ramchandran\\ Dept. of Electrical Engineering and Computer Sciences\\ University of California, Berkeley\\ \{kw1jjang, ramtin, kannanr\}@eecs.berkeley.edu}
\date{}
\maketitle

\abstract{
Group testing tackles the problem of identifying a population of $\K$ defective items from a set of $\n$ items by \emph{pooling} groups of items efficiently in order to cut down the number of tests needed. The result of a test for a group of items is \emph{positive} if any of the items in the group is defective and \emph{negative} otherwise. The goal is to judiciously group subsets of items such that defective items can be reliably recovered using the minimum number of tests, while also having a low-complexity decoding procedure.

We describe SAFFRON (\textbf{S}parse-gr\textbf{A}ph codes \textbf{F}ramework \textbf{F}or g\textbf{RO}up testi\textbf{N}g), a \emph{non-adaptive} group testing paradigm that recovers at least a $(1-\eps)$-fraction (for any arbitrarily small $\eps > 0$) of $\K$ defective items with high probability with $m=6C(\eps)K\log_2{\n}$ tests, where $C(\eps)$ is a precisely characterized constant that depends only on $\eps$. 
For instance, it can provably recover at least $(1-10^{-6})\K$ defective items with $m \simeq 68 \K \log_2{\n}$ tests. 
The computational complexity of the decoding algorithm of SAFFRON is \bigOh{\K\log\n}, which is order-optimal. 
Further, we describe a systematic methodology to robustify SAFFRON such that it can reliably recover the set of $\K$ defective items even in the presence of erroneous or \emph{noisy} test results. 
We also propose Singleton-Only-SAFFRON, a variant of SAFFRON, that recovers \emph{all} the $\K$ defective items with $m=2e(1+\alpha)\K\log\K \log_2\n$ tests with probability $1-\mathcal{O}{\left(\frac{1}{\K^\alpha}\right)}$, where $\alpha>0$ is a constant. 
By leveraging powerful design and analysis tools from modern sparse-graph coding theory, SAFFRON is the first approach to reliable, large-scale probabilistic group testing that offers both precisely characterizable  number of tests needed (down to the constants) together  with order-optimal decoding complexity.

Extensive simulation results are provided to validate the tight agreement between theory and practice.
As a concrete example, 
we simulate a case where $\K=128$ defective items have to be recovered from a population of $\n \simeq 4.3\times 10^9$ items even when up to $2\%$ of the group test results are reported wrongly (i.e., negative tests are reported as positive and vice versa). 
We run the robustified SAFFRON $1000$ times. 
We observe that \emph{all} $\K=128$ defective items are successfully recovered in every run with $\m \simeq 8.3\times 10^{5}$ tests, and the decoding time takes only about $3.8$ seconds on average on a laptop with a 2 GHz Intel Core i7 and 8 GB memory. 
}

\section{Introduction}
Group testing tackles the problem of identifying a population of $\K$ defective items from a set of $\n$ items by \emph{pooling} groups of items efficiently in order to cut down the number of tests needed. The result of a test for a group of items is \emph{positive} if any of the items in the group is defective and \emph{negative} otherwise. The goal is to judiciously group subsets of items such that defective items can be reliably recovered using the minimum number of tests, while also having a low-complexity decoding procedure.

Group testing has been studied extensively in the literature. 
Group testing arose during the Second World War \cite{hwang2000combinatorial}: 
in order to detect all soldiers infected with the syphilis virus without needing to test them individually, which was too expensive and slow, the blood samples of subsets of soldiers were pooled together and tested as groups, so that groups that tested negative would immediately exonerate all the individuals in the group.
Since then, varied theoretical aspects of group testing have been studied, and more applications of group testing have been discovered in a variety of fields spanning across biology \cite{ngo_survey}, machine learning \cite{malioutov2013exact}, medicine \cite{ganesan2015learning}, computer science \cite{goodrich_forensic}, data analysis \cite{gilbert_sparse_signal_recovery}, and signal processing \cite{olgica_compressive_sensing}. 

\subsection{Our Contributions}
The problem of group testing has been a very active area of research, and many variants have been studied in the literature.
Despite the long history of group testing, our paper has some novel intellectual and practical contributions to the field.  Our paper falls specifically into the well-studied category of large-scale probabilistic group testing where both the ambient test population size and the number of defective items are scalable, and where a targeted arbitrarily-tiny fraction of defective items can be missed.

In this work, we introduce SAFFRON (Sparse-grAph codes Framework For gROup testiNg), a powerful framework for non-adaptive group testing based on modern sparse-graph coding theory \cite{RUbook}. 
Our main intellectual contribution is that we are, to the best of our knowledge, the first to leverage the tools of sparse-graph coding theory for both group testing code design, and performance analysis, based on powerful density evolution techniques.
Recall that sparse-graph codes (e.g. Low-Density-Parity-Check (LDPC) codes \cite{RUbook}) form the backbone of reliable modern communication systems (e.g. telecommunications, wireless cellular systems, satellite and deep-space communications, etc.).  
However, pooling test design for our targeted group testing problem seems quite different from code design for the noisy communication problem, and it is not clear if the tools of the latter are even applicable here.
Concretely, 
classical coding theory deals with the design of codes based on finite-field arithmetic: e.g., in the binary-field case, this corresponds to the modulo-$2$ or XOR world.
In contrast, group testing deals with the Boolean OR world, where each observed test output is the binary OR of the states of each of the component input items in the test.  
The non-linearity of the OR operator is at odds with finite-field arithmetic, and complicates the use of classical coding theory in the group testing problem.  Our main intellectual contribution is to show how this challenge can be overcome, where we show how elegant density evolution methods and simple randomized sparse-graph coding designs can lead to powerful group test codes as well.    

A second contribution, which follows from the use of these powerful coding theory tools, is that we are able to specify {\em precise constants} in the number of tests needed while simultaneously having provable performance guarantees and order-optimal decoding complexity in the large-scale probabilistic group testing setting. 
To the best of our knowledge, this is new.

We summarize our main contributions as follows.
\begin{enumerate}[(i)]
\item The SAFFRON scheme recovers, with high probability, an arbitrarily-close-to-one fraction, $1-\eps$, of the defective items with $m=C(\eps) \K \log_2 \n$ tests, where $C(\eps)$ is a constant that depends only on $\eps$ and can be precisely computed (See Table \ref{table:c_eps}).
Moreover, the computational complexity of our decoding algorithm is $\mathcal{O}(\K \log \n)$, which is order-optimal. For instance, SAFFRON reliably recovers at least $(1-10^{-6})\K$ defective items with $m \simeq 68 \K \log_2{\n}$ tests. 
\item We propose a variant of the SAFFRON scheme, Singleton-Only-SAFFRON, which recovers \emph{all} the defective items with high probability, $1-\mathcal{O}{\left(\frac{1}{\K^\alpha}\right)}$, at the cost of $m=2e(1+\alpha)\K \log{\K} \log_2{\n}$ tests. The computational complexity of the decoding algorithm is $\mathcal{O}(\K \log \K \log \n)$. 
\item SAFFRON and its variant can be systematically robustified to noise and errors by increasing the number of tests by a constant factor that \emph{does not change the order-complexity of the scheme.} 
\end{enumerate}

Extensive simulation results are provided to validate the tight agreement between theory and practice.
As a concrete example, 
we simulate a case where $\K=128$ defective items have to be recovered from a population of $\n \simeq 4.3\times 10^9$ items with $2\%$ errors in test results.
We run the robustified SAFFRON $1000$ times. 
We observe that \emph{all} $\K=128$ defective items are successfully recovered in every run with $\m \simeq 8.3\times 10^{5}$ tests, and the decoding time takes only about $3.8$ seconds on average on a laptop with a 2 GHz Intel Core i7 and 8 GB memory. 

\subsection{Related Works} \label{sec:related_works}
We provide a brief survey of the existing results in the literature. We refer the readers to \cite{ngo2000survey, ngo_survey, mazumdar2015nonadaptive} for a detailed survey. 

We first summarize the known results on the minimum number of tests required to solve a (non-adaptive) group testing problem. 
For group testing algorithms with zero-error reconstruction, the best known lower bound on the number of required tests is $\Omega(\frac{K^2}{\log \K}{\log \n})$ \cite{d1982bounds, d1989bounds}. The best known group testing scheme under this setup requires \bigOh{\K^2 \log \n} tests \cite{hwang2000combinatorial}.
Although the zero-error reconstruction property is definitely a desired property,
such group testing schemes typically involve exhaustive table searches in their reconstruction procedures, and hence require a high computational and memory complexity of $\mathcal{O}(\K^2 \n \log \n)$ \cite{ngo_random}. 
The notable exception is a recent work \cite{indyk}, which is the first scheme that requires $\mathcal{O}(\K^2 \log \n)$ tests, while having an efficient decoding algorithm of computational complexity  $\text{poly}(\K)\cdot\mathcal{O}(\K^2 \log \n\log^2{(\K^2 \log \n)}) + \mathcal{O}(\K^4 \log^2 \n)$.

Several relaxations of the group testing problem have been studied in the literature. 
One such relaxation is allowing a small error probability as well as relaxing the requirement of perfect identification.
That is, the goal is to design a group testing scheme that identifies an approximate answer with high probability. 
Many approaches have been proposed to design group testing schemes that allow an efficient decoding algorithm for these relaxed yet important problems.
One such approach is random pooling design based on random bipartite graphs. 
For instance, randomized group testing schemes based on left-regular random bipartite graphs and right-regular random bipartite graphs are studied, respectively in \cite{bkb95} and \cite{hwa99}. 
Other lines of work have made use of existing pooling designs. 
In \cite{macula, ngo_random}, the authors use randomly chosen pools from carefully designed pools that are initially tailored for a zero-error reconstruction setting. 
With certain success probabilities, these schemes find a large fraction of the $\K$ defective items, while wrongly identifying a small fraction of normal items as defective items.
Despite the simplicity of this class of constructions, performance analysis is rather convoluted and cumbersome.
Further, these schemes are generally difficult to make robust to noise.

Another line of work is based on an information-theoretic formulation of the group testing problem.
That is, one assumes a prior distribution on the set of defective items, and searches for a group testing scheme with vanishing error probability.
In \cite{atia2012boolean}, the authors present information-theoretic bounds: with a uniform prior over the location of $\K$ defective items, $\Theta(\K \log(\frac{\n}{\K}))$ tests are necessary and sufficient (via random coding).
Mazumdar presents near-optimal explicit constructions in \cite{mazumdar2015nonadaptive}.
Chan et al. propose novel group testing algorithms and compare their performances with the fundamental lower bounds \cite{jaggi_bounds}. 
Several works have proposed group testing schemes with efficient decoding algorithms.
In \cite{gilbert_sparse_signal_recovery}, the authors present a simple group testing procedure that efficiently recovers a large fraction of $\K$ defective items with \bigOh{\K \log^2 \n} tests. 
While the proposed decoding algorithm runs in time $(\K \log\n)^{\mathcal{O}(1)}$, the algorithm returns \bigOh{\K \log n} false positives, which need to be double-checked using a $2$-stage algorithm. One of the most notable exceptions is \cite{jaggi_grotesque}. In this work, Cai et al. propose GROTESQUE, a class of efficient group testing schemes, which is the first adaptive group testing algorithm that achieves both an order-optimal number of tests and an order-optimal decoding complexity, but at the cost of using \bigOh{\log\K} adaptive stages.
Their non-adaptive scheme requires \bigOh{\K \log \K \log \n} tests, and has a decoding complexity is \bigOh{\K(\log n + \log^2 \K)}.

\subsection{Paper Organization} \label{sec:paper_organization}
The rest of the paper is organized as follows.
In Section \ref{sec:group_test}, we formally define the problem, and provide the basic ideas, based on which we develop SAFFRON.
In Section \ref{sec:SAFFRON}, we provide a detailed description of SAFFRON and its decoding algorithm. 
In Section \ref{sec:noiseless}, we provide the main theoretical results of the paper. 
In Section \ref{sec:noisy}, we robustify SAFFRON so that it can reliably recover the defective items even with some erroneous test results. 
Finally, we provide extensive simulation results in Section \ref{sec:sim}, verifying our theoretical guarantees as well as demonstrating the practical performance of SAFFRON.

\section{Group Testing Problem and Overview of the Main Results} \label{sec:group_test}

We formally define the group testing problem as follows. 
Consider a group testing problem with $\n$ items.
Among them, exactly $\K$ items are defective.
We define the support vector $\bx \in \zoset^\n$, of which the \thth{i} component is $1$ if and only if item $i$ is defective. 
That is, $x_i = \mbf{1}\{\text{item}~i~\text{is defective}\}$ for $1 \leq i \leq \n$. 
Defining $\supp{\cdot}$ as the set of indices of non-zero elements, $|\supp{\bx}| = \K$.

A subset of items can be pooled and tested, and the test result is either $1$ (positive) if any of the items in the subset is defective, or $0$ (negative) otherwise. 
For notational simplicity, we denote a subset by a binary row vector, $\ba$, of which the \thth{i} component is $1$ if and only if item $i$ belongs to the subset. 
Then, a group testing result $y$ can be expressed as 
\begin{align}
y = \langle \ba,\bx \rangle \defeq \bigvee\limits_{i=1}^n a_i x_i \label{eq:y_ax},
\end{align}
where $\vee$ is a boolean OR operator. 
Let $\m$ be the number of pools.
Denoting the \thth{i} pool by $\ba_i$, we define \emph{the group testing matrix} as $A \defeq (\ba_1^T, \ba_2^T, \ldots, \ba_\m^T)^T \in \{0,1\}^{\m\times \n}$. 
The group testing results from $\m$ pools can also be represented as a column vector $\by = (y_1, y_2, \ldots, y_\m)^T \in \{0,1\}^{m}$, where $y_i$ is the outcome of the \thth{i} test. Then, 
\begin{align}
\mathbf{y} =  A \odot \bx \defeq 
\left( \begin{array}{c}
\langle \ba_1, \bx \rangle \\
\langle \ba_2, \bx \rangle \\
\vdots \\
\langle \ba_m, \bx \rangle
\end{array}
\right).
\end{align}

The goal is to design the group testing matrix $A$ and efficiently recover the set of defective items using the $\m$ test results. 
We denote the decoding function by $g_A \colon \zoset^\m \to \zoset^\n$. 
We want the decoding result $\hat{\bx} = g_A(y)$ to be \emph{close} to the true support vector $\bx$ with some \emph{guarantee}.

Depending on different applications, one can define `closeness' and `guarantee' in different ways. Let us first discuss various notions of `closeness'. The most stringent objective is exact recovery: $\hat{\bx} = \bx$. Another objective, slightly looser, is partial recovery without false detections: one wants to make sure that $\set{\hat{\bx}} \subseteq \set{\bx}$ and $|\set{\hat{\bx}}| \geq (1-\eps) \K$. 
Another criterion is to find a superset of all the defective items without missed detections, i.e., $\set{\hat{\bx}} \supseteq \set{\bx}$ and $|\set{\hat{\bx}}| \leq (1+\eps) \K$. 

There are also several types of guarantees. The most stringent guarantee is perfect guarantee: for any $\bx$, the corresponding estimate has to be close with probability $1$. 
Another popular guarantee is probabilistic guarantee. 
That is, the decoder will provide an exact or close estimate with high probability. 

Our proposed approach, SAFFRON, recovers an arbitrarily-close-to-one fraction of the defective items (partial recovery) with high probability (probabilistic guarantee). On the other hand, we propose the variant of the SAFFRON scheme that recovers \emph{all} the defective items (perfect recovery) with high probability. 

An important generalization of the problem is to find the set of defective items with erroneous or \emph{noisy} test results. 
We show that SAFFRON can be robustified to noise, while maintaining its plain architecture.

The main results of this paper are stated in the following (informal) theorem. 

\begin{theorem}
Consider a group testing problem with $\n$ items and $\K$ defective items. SAFFRON recovers a $(1-\eps)$-fraction of the defective items for arbitrarily-close-to-zero constant $\eps$ with high probability. The number of tests and the computational complexity of the decoding algorithm are $\mathcal{O}(\K \log(\n))$, which is order-optimal in both noiseless and noisy settings. 
\end{theorem}

\begin{remark} 
In some applications, it is possible to design pools in an adaptive way. In other words, the \thth{(i+1)} pool, $\ba_{i+1}$, can be `adaptively' designed after observing the first $i$ test results or $(y_1,y_2,\ldots,y_i)$. Such a scheme is called an adaptive group testing scheme. 
While an adaptive group testing scheme requires a lower number of tests and can potentially lead to a more efficient decoding algorithm, our focus here is only on \emph{non-adaptive} group testing for the following reasons. 
First, adaptive group testing is generally applicable in limited settings of interest, whereas non-adaptive group testing applies broadly to any group testing setting.
Moreover, non-adaptive group testing enjoys an important architecture advantage.  
In contrast to adaptive group testing algorithms which are necessarily sequential in nature, non-adaptive group testing feature pre-determined pools, and can therefore be easily parallelized, leading to more efficient implementation, especially in this era of large-scale parallel computing.
\end{remark}

Table \ref{table:notation} summarizes our notation, which will be defined throughout the paper. 
\begin{table}[h]
\small
\centering
\begin{tabular}{|c|l|}
\hline
Notation      & \multicolumn{1}{c|}{Definition}                             \\ \hline
$\n$          & Number of items                                             \\ \hline
$\K$          & Number of defective items                                        \\ \hline
$\m$          & Number of pools (tests)                                            \\ \hline
$\M$          & Number of right nodes (bundles of tests)                            \\ \hline
$\bx$         & Binary representation of the set of defective items     \\ \hline
$\by$         & Binary representation of the group test results             \\ \hline
$\bz_i$       & Measurements from the \thth{i} right node                            \\ \hline
$\bz_i^j$     & Measurements from the \thth{j} stage of the \thth{i} right node        \\ \hline
$\bu_i$       & Signature vector of the \thth{i} item                         \\ \hline
$\bu_i^j$     & The \thth{j} stage of the signature vector of the \thth{i} item \\ \hline
$\bb_i$    & Binary representation of $i-1$ \\ \hline
$\be_i$       & The \thth{i} standard basis vector                             \\ \hline
$\supp{\cdot}$    & The set of the indices of the non-zero elements      \\ \hline
$w(\cdot)$    & Hamming weight or the number of ones of a vector    \\ \hline
$q$ & Probability of each test outcome being wrong       \\ \hline
$\mathcal{G}$ & Bipartite graph representation of a sparse-graph code       \\ \hline
$\overline{\mathbf{x}}$ & Bit-wise complement vector of $\mathbf{x}$ \\ \hline
$[n]$ & $\{1,2,\ldots,n\}$ \\ \hline
\end{tabular}
\caption{Summary of our notation}
\label{table:notation}
\end{table}

\section{The SAFFRON Scheme: Main Idea} \label{sec:SAFFRON}
Our test matrix design is based on an architectural philosophy that is similar to the ones in 
\cite{sameer1, sameer2, simon1, simon2, phasecode1, phasecode2, phasecode3}. 
The key idea of SAFFRON is the adoption of a design principle called `\emph{sparse signal recovery via sparse-graph codes}' that is applicable to a varied class of problems such as computing a sparse Fast Fourier Transform, a sparse Walsh Hadamard Transform, 
and the design of systems for compressive sensing and compressive phase-retrieval.
In all these problems, one designs an efficient way of sensing or \emph{measuring} an unknown sparse signal such that the decoder can estimate the unknown signal with a low decoding complexity. 
The overarching design principle is to 1) design a sensing matrix based on a sparse bipartite graph and to 2) decode the observed measurements using a simple peeling-like iterative algorithm.
We show how this same design principle allows us to efficiently tackle the group testing problem.

Consider a bipartite graph with $\n$ left nodes and $\M$ right nodes. 
Here, the $\n$ left nodes correspond to the $\n$ items, and the $\M$ right nodes corresponds to the $\M$ bundles of test results. 
We design a bipartite graph based on left-regular construction. 
That is, each left node is connected to constant number $d$ of right nodes uniformly at random.

We denote the incidence matrix of a bipartite graph $\mathcal{G}$ by $T_\mathcal{G} \in \zoset^{\M \times n}$, or simply $T$ if $\mathcal{G}$ is clear from the context. Let $\bt_i$ be the $i$th row of $T_\mathcal{G}$.
We associate each left node with a carefully designed \emph{signature} (column) vector $\bu$ of length $\h$, i.e., $\bu \in \zoset^\h$. Let us denote the signature vector of item $i$ by $\bu_i$. 
We define the signature matrix $U \defeq [u_1,u_2,\ldots,u_{n-1},u_\n ] \in \zoset^{\h \times n}$.

Given a graph $\mathcal{G}$ and a signature matrix $U$, we design our group testing matrix $A$ to be a \emph{row tensor product} of $T_{\mathcal{G}}$ and $U$, which is defined as $A = T_{\mathcal{G}} \otimes U \defeq [A_1^T, A_2^T, \ldots, A_\M^T]^T \in \zoset^{\h\M \times \n}$, where $A_i = U \text{diag}(\mathbf{t}_i)\in \zoset^{\h \times \n}$, and $\text{diag}(\cdot)$ is the diagonal matrix constructed by the input vector.
As an example, the row tensor product of matrices 
\begin{align}
T = \left [ \begin{array}{ccc}
0 & 1 & 0 \\
1 & 1 & 0 \\
0 & 0 & 1
\end{array} \right]
 ~ \text{and} ~  ~
U = \left [ \begin{array}{ccc}
1 & 0 & 1 \\
0 & 1 & 1 
\end{array} \right] 
\end{align}
is 
\begin{align}
A = T \otimes U = \left [ \begin{array}{ccc}
0 & \textbf{0} & 0 \\
0 & \textbf{1} & 0 \\ \hline
\textbf{1} & \textbf{0} & 0 \\ 
\textbf{0} & \textbf{1} & 0 \\ \hline
0 & 0 & \textbf{1} \\
0 & 0 & \textbf{1}
\end{array} \right].
\end{align}
For notational simplicity, we define the observation vector  corresponding to  right node $i$ as $\bz_i \defeq \by_{(i-1)\h+1:i\h}$ for $1 \leq i \leq \M$. Then, \begin{align}
\bz_i = U \odot \text{diag}(t_i) \bx, ~~ 1 \leq i \leq \M. \label{eq:bin_measurement}
\end{align}
In other words, $\bz_i$ is the bitwise logical ORing of all the signature vectors of the active left nodes that are connected to right node $i$.

Our decoding algorithm simply iterates through all the right node measurement vectors $\{\bz_i\}_{i=1}^{\M}$, and checks whether a right node is \emph{resolvable} or not. A right node is resolvable if exactly one new defective item can be detected by processing the right node, i.e., the location index of the defective item is found. 
The decoding algorithm is terminated when there is no more resolvable right node.

We now present the following terminologies. 
A right node that is connected to one and only one defective item is called a \emph{singleton}. A right node that is connected to two defective items is called a \emph{doubleton}.
Later, we show that with the aid of our signature matrix, 1) a singleton is resolvable, and 2) a doubleton is resolvable if one of the two defective items is already identified (in the previous iterations of the algorithm).

\subsection{Detecting and Resolving a Singleton} \label{sec:singleton_decode}
Consider the following signature matrix where the \thth{i} column is a vertical concatenation of $\bb_i$ and its complement, where $\bb_i$ is the $L$-bits binary representation of an integer $i-1$, for $i \in [\n]$, so $L = \lceil \log_2{\n} \rceil$. 
\footnote{For simplicity, the rest of paper will assume that $\n$ is a power of $2$, and hence $L = \log_2{\n}$.}
\begin{align}
\label{eq:A_0}
\left[\begin{array}{c}
U_1 \\
\overline{U}_1
\end{array}\right] &= 
\left[ \begin{array}{cccccc}
\bb_1 & \bb_2 & \bb_3 & \ldots & \bb_{\n-1} & \bb_{\n} \\ 
\rule{0pt}{3ex}    
\overline{\bb_1} & \overline{\bb_2} & \overline{\bb_3} & \ldots & \overline{\bb_{\n-1}} & \overline{\bb_{\n}}
\end{array}
\right]
= 
\left[ \begin{array}{cccccc}
0 & 0 & 0 & \ldots & 1 & 1 \\
0 & 0 & 0 & \ldots & 1 & 1 \\
\vdots & \vdots & \vdots & \ddots & \vdots & \vdots \\
0 & 0 & 1 & \ldots & 1 & 1 \\
0 & 1 & 0 & \ldots & 0 & 1 \\ \hline
1 & 1 & 1 & \ldots & 0 & 0 \\
1 & 1 & 1 & \ldots & 0 & 0 \\
\vdots & \vdots & \vdots & \ddots & \vdots & \vdots \\
1 & 1 & 0 & \ldots & 0 & 0 \\
1 & 0 & 1 & \ldots & 1 & 0
\end{array}
\right]
\end{align}
We now show that a singleton can be detected and resolved with the aid of this signature matrix.
First, note that the sum of the weight of any binary vector and the weight of its complement is always the length of the vector, $L$. 
Thus, given a singleton, the weight of the measurement vector is $L$. \emph{Furthermore, if the right node is connected to zero or more than one defective items, the weight of the measurement vector will not be $L$.} 
Therefore, by just checking the weight of the right-node measurement vector, one can simply detect whether the right node is a singleton or not. 
Further, one can also read the first half of the measurement of the detected singleton to find the index location of the defective item.

While having only $\left[\begin{array}{c} U_1 \\\overline{U}_1 \end{array}\right]$ as the signature matrix suffices for detecting and resolving singletons, in the following subsections, we show that to detect and resolve a doubleton, we need to expand the signature matrix. Thus, $\left[\begin{array}{c} U_1 \\\overline{U}_1 \end{array}\right]$ will be one part of our final signature matrix $U$.

\subsection{Resolvable Doubletons}
\label{sec:doubleton_decode}
We now design the full signature matrix $U$ by expanding $\left[\begin{array}{c} U_1 \\\overline{U}_1 \end{array}\right]$ so that one can detect and resolve both singletons and resolvable doubletons as follows. 
\begin{align}
\label{eq:A_1}
U &= 
\left[ \begin{array}{c}
U_1 \\
\overline{U}_1 \\
U_2 \\
\overline{U}_2 \\
U_3 \\
\overline{U}_3
\end{array}
\right]
=
\left[ \begin{array}{cccccc}
\bb_1 & \bb_2 & \bb_3 & \ldots & \bb_{\n-2} & \bb_{\n-1} \\ 
\rule{0pt}{3ex}    
\overline{\bb_1} & \overline{\bb_2} & \overline{\bb_3} & \ldots & \overline{\bb_{\n-1}} & \overline{\bb_{\n}} \\ \rule{0pt}{3ex}
\bb_{i_1} & \bb_{i_2} & \bb_{i_3} & \ldots & \bb_{i_{\n-1}} & \bb_{i_{\n}} \\ 
\rule{0pt}{3ex}    
\overline{\bb_{i_1}} & \overline{\bb_{i_2}} & \overline{\bb_{i_3}} & \ldots & \overline{\bb_{i_{\n-1}}} & \overline{\bb_{i_{\n}}} \\ \rule{0pt}{3ex}
\bb_{j_1} & \bb_{j_2} & \bb_{j_3} & \ldots & \bb_{j_{\n-1}} & \bb_{j_{\n}} \\ 
\rule{0pt}{3ex}    
\overline{\bb_{j_1}} & \overline{\bb_{j_2}} & \overline{\bb_{j_3}} & \ldots & \overline{\bb_{j_{\n-1}}} & \overline{\bb_{j_{\n}}}
\end{array}
\right],
\end{align}
where $\mathbf{s}_1= (i_1,i_2,\ldots,i_n)$ and $\mathbf{s}_2= (j_1,j_2,\ldots,j_n)$ are drawn uniformly at random from the set $[\n]^\n$. Then, the measurement vector for the \thth{k} right node, $\bz_k$, is 
\begin{align}
\bz_k =  U \odot \text{diag}(\bt_k) \bx  = \left[ 
\begin{array}{c}
 U_1 \odot \text{diag}(\bt_k) \bx  \\
 \overline{U}_1 \odot \text{diag}(\bt_k) \bx  \\
 U_2 \odot \text{diag}(\bt_k) \bx  \\
 \overline{U}_2 \odot \text{diag}(\bt_k) \bx  \\
 U_3 \odot \text{diag}(\bt_k) \bx  \\
 \overline{U}_3 \odot \text{diag}(\bt_k) \bx  
\end{array}
\right]
\defeq \left[ 
\begin{array}{c}
\bz^1_k\\
\bz^2_k\\
\bz^3_k\\
\bz^4_k\\
\bz^5_k\\
\bz^6_k
\end{array}
\right]. \label{eq:right_node_measurement}
\end{align}
We call $\bz_k^w$ the \thth{w} section of the \thth{k} right-node measurement vector. Figure \ref{fig:z_k} provides an illustration of a right-node measurement vector $\bz_k$. 
Assume the support vector is $\bx = (0,1,1,\ldots,0,1,0)$. That is, there are exactly $\K = 3$ defective items: item $2$, item $3$ and item $n-1$. 
Also, assume that the \thth{k} right node is connected to item $2$, item $n-2$ and item $n-1$, i.e., $\bt_k = (0,1,0,\ldots,1,1,0)$. 
Consider the \thth{k} right node and its corresponding right-node measurement vector $\bz_k$. 
By \eqref{eq:right_node_measurement}, $\bz_k$ is equal to $\bu_2 + \bu_{n-1}$, and it consists of $6$ sections. We can also find each section of $\bz_k$ by looking at the corresponding sections of $\bu_2$ and $\bu_{n-1}$, as described in Figure \ref{fig:z_k}. For instance, the first section of the \thth{k} right-node measurement vector $\bz_k^1$ is equal to $\bb_2+ \bb_{n-1}$.

\begin{figure}[t]
    \centering
    \includegraphics[width=0.35\textwidth]{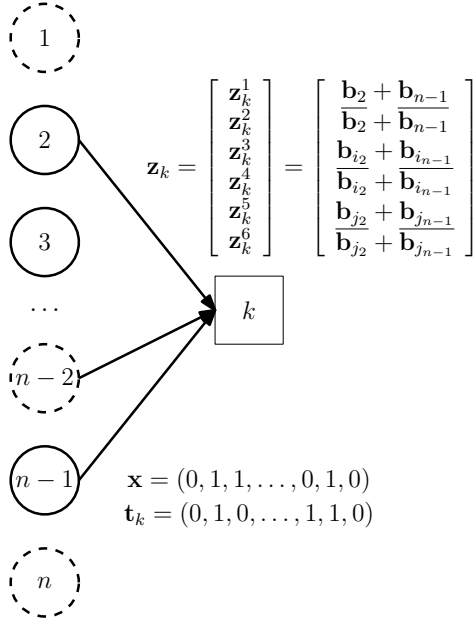}
    \caption{\textbf{An illustration of a right-node measurement vector $\bz_k$ and its $6$ sections}. When the support vector $\bx = (0,1,1,\ldots,0,1,0)$ and the \thth{k} row of the bipartite graph $\bt_k = (0,1,0,\ldots,1,1,0)$ are given, the \thth{k} right-node measurement vector $\bz_k$ is equal to $\bu_2 + \bu_{n-1}$, as depicted in the figure. The right-node measurement vector consists of sections $\{\bz^i_k\}_{i=1}^{6}$.}
    \label{fig:z_k}
\end{figure}

Assume that right node $k$ is connected to exactly one identified defective item, say $\ell_0$. 
The decoder first \emph{guesses} that the right node is a resolvable doubleton. That is, the right node is connected to exactly two defective items: one of them is the identified defective item $\ell_0$, and the other is the unidentified defective item $\ell_1$. Thus,
\begin{align}
\left[\begin{array}{c}
\bz^1_k \\
\bz^2_k 
\end{array}\right]
= \bu_{\ell_0} \vee \bu_{\ell_1}
=  
\left[ \begin{array}{c}
\bb_{\ell_0} \\ \rule{0pt}{3ex}  
\ol{\bb_{\ell_0}} \end{array} \right] 
\vee 
\left[ \begin{array}{c}
\bb_{\ell_1} \\ \rule{0pt}{3ex}  
\ol{\bb_{\ell_1}}
\end{array} \right]
=  
\left[ \begin{array}{c}
\bb_{\ell_0} \vee \bb_{\ell_1}  \\ \rule{0pt}{3ex}  
\ol{\bb_{\ell_0}} \vee \ol{\bb_{\ell_1}} \end{array} \right].
\end{align}
One can always recover any bit of $\bb_{\ell_1}$ as follows. Consider the first bit of $\bb_{\ell_1}$. If $\bb_{\ell_0,1} = 0$, $\bb_{\ell_1,1} = \bz^1_{k,1}$. If not, one can read the first bit from the second section and invert it, i.e., $\bb_{\ell_1,1} = \ol{\bz^2_{k,1}}$. 
Therefore, if the assumption is true, the decoder is able to recover the other defective item's index, $\ell_1$.
Similarly, the decoder applies the decoding algorithm to the other four segments from $\bz_k^3$ to $\bz_k^6$, and obtain two other indices, $\ell_2$ and $\ell_3$. 
Finally, the decoder \emph{checks} whether $\ell_2 = i_{\ell_1}$ and whether $\ell_3 = j_{\ell_1}$. If any of these conditions does not hold, the decoder concludes that the guess is wrong. If both conditions hold, the decoder concludes that the hypothesis is correct, declares the right node to be resolved, and declares a new defective item of index $\ell_1$. 
Then,
\begin{lemma} \label{lemma:noiseless}
SAFFRON successfully detects and resolves all the singletons and all the resolvable doubletons. When a right-node measurement vector is connected to more than $2$ defective items, SAFFRON declares a wrong defective item with probability no greater than $\frac{1}{\n^2}$. 
\end{lemma} 
\begin{proof}
From the above decoding algorithm description, it is clear that SAFFRON correctly detects and resolves all the singletons and resolvable doubletons. 
Consider a right node connected to more than $2$ defective items. From the first two sections of the corresponding right-node measurement vector, SAFFRON reads the first index $\ell_1$.  
SAFFRON declares a wrong defective item only if item $\ell_1$ is not connected to the right node and the two check equations hold. 
We note that when item $\ell_1$ is not connected to the right node, $i_{\ell_1}$ and $j_{\ell_1}$ are independent of $\ell_2$ and $\ell_3$. 
Thus, the probability that $\ell_2 = i_{\ell_1}$ and $\ell_3 = j_{\ell_1}$ is $\frac{1}{n^2}$. 
\end{proof}

\begin{remark} \label{remark:4sections}
Lemma \ref{lemma:noiseless} ensures that the signature matrix $U$ in \eqref{eq:A_1} can be used to detect and resolve all the singletons and all the resolvable doubletons with error probability no greater than $\frac{1}{\n^2}$.
Later, we show that SAFFRON performs the described right-node decoding algorithm \bigOh{\K} times in total.
Thus, by the union bound, the probability of having any error in \bigOh{\K} tests is bounded by \bigOh{\frac{\K}{\n^2}}, which always approaches zero as $\n$ increases. 
However, if the number of defective items is sublinear in $\n$, to reduce the number of tests, one can keep only the first $4$ sections in the signature matrix $U$. 
By doing so, the probability of having any error in \bigOh{\K} tests is upper bounded by \bigOh{\frac{\K}{\n}}, which still vanishes since $\K = o(\n)$. 
Therefore, one can always use a shorter signature matrix in order to save the number of tests by $33\%$, when the number of defective items is sublinear in $n$.
\end{remark}

\subsection{Example}
In this section, we provide an illustrative example of the decoding algorithm of SAFFRON. 

Consider a group testing problem with $\n = 8$ items and $\K = 3$ defective items. 
Let $\bx = (1, 0, 1, 0, 0, 0, 0, 1)$, i.e., item $1$, item $3$ and item $8$ are defective items. 
We show how SAFFRON can find the set of defective items.
Recall that we design our group testing matrix $A$ to be a row tensor product of $T_\mathcal{G}$ and $U$. 
Assume that a bipartite graph $\mathcal{G}$ is designed as follows. 
\footnote{
In the interest of conceptual clarity of the toy example, here we present a bipartite graph that is not left-regular. 
}
\begin{align*}
T_\mathcal{G} = \left [ \begin{array}{cccccccc}
0 & 1 & 1 & 1 & 0 & 1 & 0 & 0 \\
1 & 1 & 1 & 1 & 0 & 0 & 1 & 1 \\
1 & 0 & 0 & 0 & 1 & 0 & 1 & 1 \\
0 & 1 & 1 & 0 & 1 & 1 & 0 & 1 
\end{array} \right] \in \zoset^{\M \times \n}
\end{align*}
We have $\M = 5$ right nodes, and $\n = 8$ items are connected to them according to $T_\mathcal{G}$.

Assume that we drew random sequences $\mathbf{s}_1$ and $\mathbf{s}_2$, and the realization of them are as follows.
\begin{align}
\mathbf{s}_1 &= (5,2,4,8,7,1,3,6), ~~\mathbf{s}_2 = (3,1,5,6,3,8,2,7)
\end{align}
Thus, the measurement matrix of SAFFRON is as follows.
\begin{align}
\label{eq:A_1}
U &= 
\left[ \begin{array}{c}
U_1 \\
\overline{U}_1 \\
U_2 \\
\overline{U}_2 \\
U_3 \\
\overline{U}_3  
\end{array}
\right]
= 
\left[ \begin{array}{cccccccc}
0 & 0 & 0 & 0 & 1 & 1 & 1 & 1 \\
0 & 0 & 1 & 1 & 0 & 0 & 1 & 1 \\
0 & 1 & 0 & 1 & 0 & 1 & 0 & 1 \\ \hline
1 & 1 & 1 & 1 & 0 & 0 & 0 & 0 \\
1 & 1 & 0 & 0 & 1 & 1 & 0 & 0 \\
1 & 0 & 1 & 0 & 1 & 0 & 1 & 0 \\ \hline \hline
1 & 0 & 0 & 1 & 1 & 0 & 0 & 1 \\
0 & 0 & 1 & 1 & 1 & 0 & 1 & 0 \\
0 & 1 & 1 & 1 & 0 & 0 & 0 & 1 \\ \hline
0 & 1 & 1 & 0 & 0 & 1 & 1 & 0 \\
1 & 1 & 0 & 0 & 0 & 1 & 0 & 1 \\
1 & 0 & 0 & 0 & 1 & 1 & 1 & 0 \\ \hline \hline
0 & 0 & 1 & 1 & 0 & 1 & 0 & 1 \\
1 & 0 & 0 & 0 & 1 & 1 & 0 & 1 \\
0 & 0 & 0 & 1 & 0 & 1 & 1 & 0 \\ \hline
1 & 1 & 0 & 0 & 1 & 0 & 1 & 0 \\
0 & 1 & 1 & 1 & 0 & 0 & 1 & 0 \\
1 & 1 & 1 & 0 & 1 & 0 & 0 & 1
\end{array} 
\right]
= \left[ \bu_1, \bu_2, \bu_3, \bu_4, \bu_5, \bu_6, \bu_7, \bu_8 \right]
\end{align}

Using (\ref{eq:bin_measurement}), we have the following equations for the \thth{4} right-node measurement vectors.
\begin{align}
\by = \left(\begin{array}{c}
\bz_1 \\ 
\bz_2 \\ 
\bz_3 \\
\bz_4 
\end{array} \right) 
= \left(\begin{array}{c}
\bu_3 \\
\bu_1 \vee \bu_3 \vee \bu_8 \\
\bu_1 \vee \bu_8 \\ 
\bu_3 \vee \bu_8 
\end{array} \right)
\end{align}

Thus, we will observe the following right-node measurement vectors.
\begin{align}
\bz_1 = \left(0,1,0,1,0,1,0,1,1,1,0,0,1,0,0,0,1,1\right)^T\\
\bz_2 = \left(1,1,1,1,1,1,1,1,1,1,1,1,1,1,0,1,1,1\right)^T\\
\bz_3 = \left(1,1,1,1,1,1,1,0,1,0,1,1,1,1,0,1,0,1\right)^T\\
\bz_4 = \left(1,1,1,1,0,1,1,1,1,1,1,0,1,1,0,0,1,1\right)^T
\end{align}

We are now ready to decode these measurements. The decoding algorithm first finds all the singletons by checking whether a right-node measurement's weight is $3L = 3\log_2{\n}$. The weights are as follows.
\begin{align}
w(\bz_1) = 9,~~w(\bz_2) = 17,~~w(\bz_3) = 14,~~w(\bz_4) = 14
\end{align}
Since $w(\bz_1) = 3\log_2{\n}$, the decoder declares that right node $1$ is a singleton. Then, it can read off the first $3$ bits of $\bz_1$. As $\bz^1_1 = (0,1,0)$, the decoder concludes that item $3$ is defective.

In the second iteration, the algorithm inspects right nodes that are potentially resolvable doubletons including defective item $3$. Since $T_{2,1} = T_{4,1} = 1$, right nodes $2$ and $4$ are inspected. 

Consider right node $2$. 
We hypothesize that the right node is a doubleton consisting of defective item $3$ and exactly one other unknown defective item. That is, we guess that $\bz_2 = \bu_3 \vee \bu_{\ell_1}$, and recover $\ell_1,\ell_2$ and $\ell_3$ as described in the previous section. Then,
\begin{align}
\ell_1 = 6,~~\ell_2=8,~~\ell_1=3
\end{align}
By noticing that $i_{\ell_1} = i_6 = 1 \neq {\ell_2},~j_{\ell_1} = j_6 = 8 \neq {\ell_3}$, the decoder declares that the right node is \emph{not} a resolvable doubleton that contains item $3$. 

Consider right node $4$. The decoder again makes a guess that $\bz_4 = \bu_3 \vee \bu_{\ell_1}$. Then, it obtains three indices as follows.
\begin{align}
\ell_{1} = 8,~~\ell_{2}=6,~~\ell_{3}=7
\end{align}
By noticing that $i_{\ell_{1}} = i_8 = 6 = \ell_{2},~j_{\ell_{1}} = j_8 = 7 = \ell_{3}$, the decoder declares that right node $4$ is a resolvable doubleton including item $3$. Moreover, it also finds that the other defective item's index is $\ell_1 = 8$. 

In the third iteration, the decoder knows that right node $2$ is not resolvable anymore as it already includes two identified defective items. However, right node $3$ now has a possibility of being a resolvable doubleton as the decoder found defective item $8$ in the previous iteration, and defective item $8$ is also in right node $3$. 
The decoder hypothesizes that right node $3$ is a doubleton, i.e., $\bz_3 = \bu_8 \vee \bu_{\ell_1}$. Similarly, the decoder reads three indices, and the recovered three indices are as follows.
\begin{align}
\ell_1 = 1,~~\ell_2=5,~~\ell_3=3
\end{align}
Because $i_{\ell_1} = i_1 = 5 = \ell_2,~j_{\ell_1} = j_1 = 3 = \ell_3$, the decoder will conclude that right node $3$ is a doubleton including defective item $8$, and that the other defective item's index is $\ell_1 = 1$. 

The algorithm is terminated as there are no more right nodes to be resolved, concluding that items $1$, $3$ and $8$ are defective items.

\section{Main Results} \label{sec:noiseless}

In this section, we analyze the SAFFRON scheme. The main theoretical result of this paper is the following theorem.
\begin{theorem} 
With $m = 6C(\eps)\K \log_2{\n}$ tests, SAFFRON recovers at least $(1-\eps)\nod$ defective items with probability $1-$\bigOh{\frac{\K}{\n^2}}, where $\eps$ is an arbitrarily-close-to-zero constant, and $C(\eps)$ is a constant that depends only on $\eps$. Table \ref{table:c_eps} shows some pairs of $\eps$ and $C(\eps)$. 
\begin{table}[h]
\centering
\begin{tabular}{|l|l|l|l|l|l|l|l|l|}
\hline
Error floor, $\eps$    & $10^{-3}$       & $10^{-4}$       & $10^{-5}$       & $10^{-6}$        & $10^{-7}$        & $10^{-8}$        & $10^{-9}$        & $10^{-10}$       \\ \hline
$C(\eps) = \frac{d^{\star}}{\la^{\star}} $             &  \textbf{$6.13$} & \textbf{$7.88$} & \textbf{$9.63$} & \textbf{$11.36$} & \textbf{$13.10$} & \textbf{$14.84$} & \textbf{$16.57$} & \textbf{$18.30$} \\ \hline
Left-deg, $d^{\star}$                & $7$             & $9$             & $10$            & $12$             & $14$             & $15$             & $17$             & $19$             \\ \hline
\end{tabular}
\caption{Pairs of $\eps$ and $C(\eps)$}
\label{table:c_eps}
\end{table}

The computational complexity of the decoding algorithm is linear in the number of measurements, i.e., \bigOh{\K \log{\n}}, that is order-optimal.
\label{thm:singleton_doubleton}
\end{theorem}

\begin{proof}
First note that each right node is associated with $6\log_2 \n$ tests based on \eqref{eq:bin_measurement}. Thus, we only need to show that the number of required right nodes to guarantee successful completion of the algorithm is $C(\eps)\K$.

We design a $d$-left-regular bipartite graph with $\n$ left nodes and $\M$ right nodes as follows. Each left node is connected to a set of $d$ right nodes uniformly at random, independently from other left nodes.
For the analysis, we focus on the \emph{pruned} bipartite graph constructed by the $\K$ defective left nodes and the right nodes. 
Then, the average right degree $\la=\frac{\K d}{\M}$. 
Further, as $\K$ gets large, the degree distribution of right nodes approaches a Poisson distribution with parameter $\la$. 
We define the right edge-degree distribution $\rho(x) = \sum_{i=1}^{\infty}{\rho_i x^{i-1}}$, where $\rho_i$ is the probability that a randomly selected edge in the graph is connected to a right node of degree $i$. Then, with this design of random bipartite graphs, 
\begin{align}
\rho_i = \frac{i\M}{\K d} \Pr(\text{degree of a random right node} = i) = \frac{i\M}{\K d} e^{-\la} \frac{\la^i}{i!} = e^{-\la} \frac{\la^{i-1}}{(i-1)!}.
\end{align}
Hence, $\rho(x) = e^{-\la(1-x)}$.

SAFFRON performs an iterative decoding procedure as follows. 
In the first round, it finds all the singletons and their corresponding defective items as described in Section \ref{sec:singleton_decode}. In the following rounds, it detects and resolves all the resolvable doubletons and recovers their corresponding defective items. This process is repeated until no new defective items are recovered during one iteration.

The fraction of defective items that cannot be identified at the end of this iterative decoding algorithm can be analyzed by density evolution \cite{RUbook, shokrollahi2004ldpc}. 
Density evolution is a tool to analyze a message-passing algorithm. 
At iteration $j$ of the algorithm, an unidentified defective item passes a message to its neighbor right nodes that it has not been recovered. 
Let $p_{j}$ be the probability that a random defective item is not identified at iteration $j$. 
The density evolution relates $p_j$ to $p_{j+1}$ as follows. 
\begin{align}
p_{j+1} = \Pr(\text{not resolvable from one children right node})^{d-1} = \left[ 1 - \left( \rho_1 + \rho_2 (1 - p_j) \right) \right] ^ {d-1},
\end{align}
where $\rho_1 = e^{-\la}$, and $\rho_2 = \la e^{-\la}$. To prove the above equation, consider the graph shown in Figure \ref{fig:tree}. 
\begin{figure}[h]
    \centering
    \includegraphics[width=0.4\textwidth]{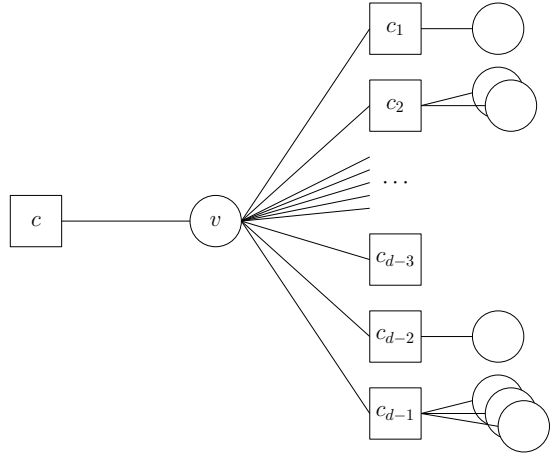}
    \caption{\textbf{A tree-like neighborhood of an edge between left node $v$ and right node $c$}. At iteration $j+1$, a `not-recovered' message is passed through this edge if and only if none of the other neighbors of $v$, $\{c_i\}_{i=1}^{d-1}$, have been identified as either a singleton or a resolvable doubleton at iteration $j$. 
}
    \label{fig:tree}
\end{figure} 
At iteration $j+1$, left node $v$ passes a `not-recovered' message to right node $c$ if none of its other neighbor right nodes $\{c_i\}_{i=1}^{d-1}$ has been identified as either a singleton or a resolvable doubleton at iteration $j$. 
The probability that a particular neighbor right node, say $c_1$, has been resolved (either as a singleton or doubleton) at iteration $j$ is $\rho_1 + \rho_2 (1 - p_j)$. Then, given a tree-like neighborhood of $v$, all the messages are independent. Thus, the above equation follows.
\begin{figure}[h]
    \centering    \includegraphics[width=0.5\textwidth]{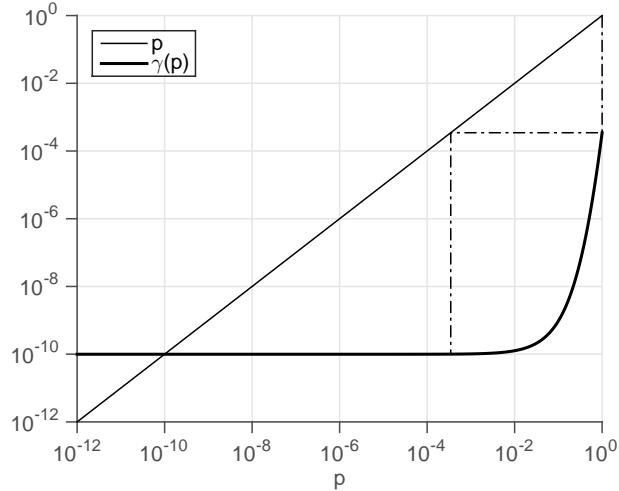}
    \caption{\textbf{Density evolution.} 
    Plotted is the function $\gamma(p) =\left[ 1 - \left( \rho_1 + \rho_2 (1 - p) \right) \right]^{d-1}$ with $\la = 1.038$ and $d = 19$ along with $y=p$. Note that $\gamma(p)$ meets the identity function at $p = 10^{-10}$. The dash-dotted line shows how $\{p_j\}$ evolves starting from $p_j = 1$ as $j$ increases.
    }
    \label{fig:density_evo}
\end{figure}

To characterize the fraction of defective items that will not be recovered by the time the algorithm terminates, we find the limit of the sequence $\{p_j\}$ as $j \rightarrow \infty$. 
Let $\eps = \lim_{j\rightarrow \infty}{ p_j}$. One can approximately calculate 
\begin{align}
\eps = \left[ 1 - ( \rho_1 + \rho_2) \right] ^ {d-1} = \left[ 1 - e^{-\la} - \la e^{-\la} \right] ^ {d-1}.
\label{eq:eps}
\end{align}

We can also pictorially see how $p_j$ evolves as follows. 
We first draw a vertical line from $(p_j,p_j)$ to $(p_j,\gamma(p))$ and then a horizontal line from $(p_j,\gamma(p_j))$ to $(\gamma(p_j),\gamma(p_j)) = (p_{j+1}, p_{j+1})$. We repeat the above procedure until $p_j$ converges to a fixed point. 
Figure \ref{fig:density_evo} plots the function $\gamma(p) \defeq \left[ 1 - \left( \rho_1 + \rho_2 (1 - p) \right) \right] ^ {d-1}$ when $\la = 1.038$ and $d = 19$. 
Starting from $p_1 = 1$, we observe that $p_j$ converges to $\eps =  \left[ 1 - e^{-\la} - \la e^{-\la} \right] ^ {d-1} \simeq 10^{-10}$.

We now find a pair of design parameters $(d, \M)$ that minimize the number of right nodes $\M$ (thus the number of tests) given a targeted reliability $\eps$.
\begin{align}
(d-1) \log (1 - e^{-\la} -\la e^{-\la}) =\log \eps.
\label{eq:n_of_bins_doubleton}
\end{align}

We solve the following optimization problem numerically.
\begin{align}
\underset{\la > 0}{\text{min}} &~~~~ M = \frac{\K d}{\la} \\
\text{subject to} &~~~~ (d-1) \log (1 - e^{-\la} -\la e^{-\la}) =\log \eps.
\end{align}

We numerically solve the optimization problem and attain the optimal $\la^{\star}$ and $d^{\star}$ as a function of $\eps$. Some of the optimal design parameters for different reliability levels are shown in Table \ref{table:c_eps}.
We define $C(\eps) = \frac{\M}{\K} = \frac{d^{\star}}{\la^{\star}}$.

Up to now, we have analyzed the average fraction of unidentified defective items over a randomly constructed bipartite graph. To complete the proof, we need to show three more steps. First, we need to show that after a fixed number of iterations, $p_j$ gets arbitrarily close to $\eps$. Second, we need to show that with high probability, a constant-depth neighborhood of a random left node is a tree. Finally, using the steps $1$ and $2$, we show that the actual fraction of unidentified defective items is highly concentrated around its average. These proofs are identical to the ones (Corollary 2.5., Lemma 2.6., Lemma 2.7.) in \cite{phasecode_arxiv}, which we omit for the purpose of readability.
\end{proof}

\subsection{A Variant of SAFFRON: Singleton-Only-SAFFRON}
We now present a variant of the SAFFRON scheme that only detects and resolves singletons. 
We call this scheme \emph{Singleton-Only-SAFFRON}. 
Clearly, the decoding algorithm of Singleton-Only-SAFFRON is not iterative; instead, it detects and resolves all singletons in a single-stage procedure.
We remark that Singleton-Only-SAFFRON is closely related to the non-adaptive GROTESQUE, proposed in \cite{jaggi_grotesque}: both schemes detect and resolve singletons only. 
However, Singleton-Only-SAFFRON requires significantly lower number of tests than the non-adaptive GROTESQUE.
This is because our deterministic signature matrix is more efficient than the random signature matrices.

We show that using only singletons costs us an extra factor $\mathcal{O}(\log \K)$ in the number of tests and computational complexity. 
However, the Singleton-Only-SAFFRON scheme can recover \emph{all} the $\K$ defective items with high probability. 

The measurement matrix of Singleton-Only-SAFFRON is similar to the one of SAFFRON with the difference that the signature matrix only consists of $U_1$ and $\overline{U}_1$ as stated in \eqref{eq:A_0}. The reason is that the algorithm does not intend to detect and resolve doubletons, as explained in Section \ref{sec:doubleton_decode}.

\begin{theorem}
With $m = 2e(1+\alpha)\nod \log{\nod} \log_2{\noi} \simeq 5.437(1+\alpha) \nod \log{\nod} \log_2{\noi}$ tests, Singleton-Only-SAFFRON finds \emph{all} the $\nod$ defective items with probability $1-$\bigOh{\frac{1}{K^\alpha}}, where $e$ is the base of the natural logarithm, and $\alpha > 0$. The computational complexity of the decoding algorithm is linear in the number of measurements, i.e., \bigOh{\nod \log{\nod} \log{\noi}}.
\label{thm:singleton_2}
\end{theorem}

\begin{proof}
First note that each right node is associated with $2\log_2 \n$ tests based on \eqref{eq:A_0}. Thus, we only need to show that the number of required right nodes to guarantee successful completion of the algorithm is $e(1+\alpha)\K\log\K$.

We design a bipartite graph with $M = e(1+\alpha)\nod \log{\nod}$ right nodes. Each left node is connected to a certain right node with probability $p = \frac{1}{\nod}$ independently of everything else. 
Then, the average degree of a right node is $\la = Kp = 1$. Note that the probability that a certain defective item is not connected to any singletons is as follows.
\begin{align}
\sum_{i=0}^{M}{{M \choose i}p^i (1-p)^{M-i} (1 - e^{-1})^i} = (p(1-e^{-1}) + (1-p))^M = (1-pe^{-1})^M
\end{align}
Hence, the probability that any of the $\K$ items is not found can be bounded using union bound as follows.
\begin{align}
P_e &\leq \nod (1-pe^{-1})^M  \\
&\leq \nod \left(e^{ -pe^{-1} }\right)^{(1+\alpha)e\nod \log \nod} \label{eq:ineq1} \\
&= \nod^{-\alpha }
\end{align}
In (\ref{eq:ineq1}), we use $1-x \leq e^{-x}$ for all $x$. 
\end{proof}

\section{Robustified SAFFRON for Noisy Group Testing} \label{sec:noisy}
In this section, we robustify SAFFRON such that it can recover the set of $\K$ defective items with erroneous or \emph{noisy} test results. 
We assume an i.i.d. noise model. That is, each test result is `wrong' with probability $q$, i.e.,
\begin{align}
\mathbf{y} =  A\odot\bx + \bw,
\end{align}
where the addition is over binary field, and $\bw$ is an i.i.d. noise vector whose components are $1$ with probability $0 < q < \frac{1}{2}$ and $0$ otherwise. \footnote{If $q > \frac{1}{2}$, one can always take the complement of all the test results, and treat the channel as if each test result is wrong with crossover probability $0<\tilde{q} = 1-q < \frac{1}{2}$.}

Our approach is simple: we design the robust signature matrix $U'$ consisting of \emph{encoded} columns of $U_j$ for $1 \leq j \leq 3$ and their complements.
Intuitively, we treat each column of $U_j$ as a message that needs to be transmitted over a noisy memoryless communication channel. An efficient modern error-correcting code guarantees reliable decoding of the signature. 
Spatially-coupled LDPC codes have the following properties \cite{spatial_LDPC}.
\begin{itemize}
\item It has an encoding function $f(\cdot): \zoset^{N} \rightarrow \zoset^{N/R}$ and a decoding function $g(\cdot): \zoset^{N/R} \rightarrow \zoset^{N}$, and its decoding complexity is \bigOh{N}.
\item If $R$ satisfies
\begin{align}
R < 1 - H(q) - \delta = q\log_2{q} + (1-q)\log_2{(1-q)} -\delta,
\end{align} 
for an arbitrarily small constant $\delta > 0$, then $\Pr(g(\bx + \bw) \neq \bx) < 2^{-\zeta N}$ as $N$ approaches infinity, for some constant $\zeta > 0$.
\end{itemize} 

Note that each column of $U_j$ is of length $\log_2 \n$. Thus, $N = \log_2 \n$ in this setup.
Given such an error-correcting code, we design the signature matrix $U$ for the robust SAFFRON scheme as follows.
\begin{align}
U &= 
\left[ \begin{array}{cccccc}
f(\bb_1) & f(\bb_2) & f(\bb_3) & \ldots & f(\bb_{n-2}) & f(\bb_{n-1}) \\ \rule{0pt}{3ex}   
\ol{f(\bb_1)} & \ol{f(\bb_2)} & \ol{f(\bb_3)} & \ldots & \ol{f(\bb_{n-2})} & \ol{f(\bb_{n-1})} \\ \rule{0pt}{3ex}   
f(\bb_{i_1}) & f(\bb_{i_2}) & f(\bb_{i_3}) & \ldots & f(\bb_{i_{\n-1}}) & f(\bb_{i_{\n}}) \\ \rule{0pt}{3ex}   
\ol{f(\bb_{i_1})} & \ol{f(\bb_{i_2})} & \ol{f(\bb_{i_3})} & \ldots & \ol{f(\bb_{i_{\n-1}})} & \ol{f(\bb_{i_{\n}})} \\ \rule{0pt}{3ex}   
f(\bb_{j_1}) & f(\bb_{j_2}) & f(\bb_{j_3}) & \ldots & f(\bb_{j_{\n-1}}) & f(\bb_{j_{\n}}) \\ \rule{0pt}{3ex}   
\ol{f(\bb_{j_1})} & \ol{f(\bb_{j_2})} & \ol{f(\bb_{j_3})} & \ldots & \ol{f(\bb_{j_{\n-1}})} & \ol{f(\bb_{j_{\n}})}
\end{array}
\right]
\in \{0,1\}^{\frac{6\log_2 n}{R}\times n}
\end{align}

We now describe how the robustified SAFFRON detects and resolves a singleton.
Consider a singleton right node $k$. Then, the right-node measurement vector is of the following form.
\begin{align}
\bz_k = \bu_k + \bw
\end{align}
The decoder first applies the decoding function $g(\cdot)$ to the first, third, and fifth segments of the right-node measurement vector and obtains $g(\bz^1_k), g(\bz^3_k), $ and $g(\bz^5_k)$. Let $\ell_1 - 1$ be the decimal representation of $g(\bz^{1}_k)$, $\ell_2 -1$ be the decimal representation of $g(\bz^{3}_k)$, and $\ell_3-1$ be the decimal representation of $g(\bz^{5}_k)$.
We use the following singleton detection rule. 
\begin{align} \label{eq:permute_check}
i_{\ell_1} = \ell_2 ,~~ j_{\ell_1} = \ell_3 
\end{align}
If the above conditions are satisfied, the decoder declares a singleton and the location index of the defective item $\ell_1$. 
\begin{lemma} \label{lem:noisy_singleton_singleton}
Robustified-SAFFRON misses a singleton with probability no greater than $\frac{3}{\n^\zeta}$. 
Robustified-SAFFRON wrongly declares a defective item with probability no greater than $\frac{1}{\n^{2+\zeta}}$.
\end{lemma}
\begin{proof}
Robustified-SAFFRON misses a singleton only if any of the $3$ decoded indices is wrong. The probability of such an event is upper bounded by $\frac{3}{\n^\zeta}$ by the union bound. 
Robustified-SAFFRON wrongly declares a defective item only if $\ell_1$ is wrongly decoded, but \eqref{eq:permute_check} still holds. Such an event happens with probability no greater than $\frac{1}{\n^{2+\zeta}}$.  \end{proof}
Lemma \ref{lem:noisy_singleton_singleton} implies that the robustified SAFFRON scheme will miss fewer than $\frac{3}{\n^\zeta}$-fraction of singletons. We compensate this loss by increasing the number of right nodes: instead of using $M$ right nodes, we use $M \left(1+\frac{3}{\n^\zeta} \right)$ right nodes, so that the effective number of right nodes becomes $M \left(1+\frac{3}{\n^\zeta} \right) \left(1-\frac{3}{\n^\zeta} \right) \simeq M$.

Now, consider right node $k $ that is a resolvable doubleton with an identified defective item $\ell_0$ and an unidentified defective item $\ell_1$.  
Then, the right-node measurement vector is of the following form.
\begin{align}
\left[\begin{array}{c}
\bz^1_k \\
\bz^2_k \\
\bz^3_k \\
\bz^4_k \\
\bz^5_k \\
\bz^6_k 
\end{array}\right]
= \bu_{\ell_0} \vee \bu_{\ell_1} + \bw_k
=  
\left[ \begin{array}{c}
f(\bb_{\ell_0}) \\ \rule{0pt}{3ex}  
\ol{f(\bb_{\ell_0})} \\
f(\bb_{i_{\ell_0}}) \\ \rule{0pt}{3ex}  
\ol{f(\bb_{i_{\ell_0}})} \\
f(\bb_{j_{\ell_0}}) \\ \rule{0pt}{3ex}  
\ol{f(\bb_{j_{\ell_0}})} 
\end{array} \right] 
\vee 
\left[ \begin{array}{c}
f(\bb_{\ell_1}) \\ \rule{0pt}{3ex}  
\ol{f(\bb_{\ell_1})} \\
f(\bb_{i_{\ell_1}}) \\ \rule{0pt}{3ex}  
\ol{f(\bb_{i_{\ell_1}})} \\
f(\bb_{j_{\ell_1}}) \\ \rule{0pt}{3ex}  
\ol{f(\bb_{j_{\ell_1}})} 
\end{array} \right]
+
\left[ \begin{array}{c}
\bw_k^1\\
\bw_k^2\\
\bw_k^3\\
\bw_k^4\\
\bw_k^5\\
\bw_k^6
\end{array} \right]
\end{align}

Consider the first two sections of the measurement vector $\bz^1_k$ and $\bz^2_k$. 
We show that since the decoder knows $\bu_{\ell_0}$, it can get access to the measurement $\bu_{\ell_1} + \bw_k$. 
To this end, the decoder first looks at the first bit of $\bu_{\ell_0}$. If this bit is $0$, then the first bit of $\bz^1_k$ is indeed the first bit of $f(\bb_{\ell_1}) + \bw^1_k$. Now, if the first bit of $\bu_{\ell_0}$ is $1$, the decoder complements the first bit of $\bz^2_k$ to get the first bit of $f(\bb_{\ell_1}) + \bw^2_k$. Similar procedure can be performed for all bits. By collecting all the bits, the decoder has access to $f(\bb_{\ell_1}) + \bw'^1_k$, where $\bw'^1_k$ has the same statistics as the original i.i.d. noise. Then, it can apply the decoding function to obtain $\ell_1$. 
Similarly, using the other sections of the measurement vector, the decoder can recover $\ell_2$ and $\ell_3$. Thus, by testing the detection rules in \eqref{eq:permute_check}, the decoder can detect a resolvable doubleton. 

We now present the following theorem for noisy group testing. 

\begin{theorem}
With $m = 6\beta(q)C(\eps)\K \log_2{\n}$ tests, Robustified-SAFFRON can recover at least $(1-\eps)\nod$ defective items with probability $1-$\bigOh{\frac{\K}{\n^{2+\zeta}}}, where $\eps$ is an arbitrarily-close-to-zero constant, $C(\eps)$ is a constant that depends only on $\eps$,
and $\beta(q) = \frac{1}{R} > \frac{1}{1-H(q)-\delta}$ for an arbitrarily small constant $\delta > 0$. Table \ref{table:c_eps} shows some pairs of $\eps$ and $C(\eps)$. The computational complexity of the decoding algorithm is linear in the number of tests, i.e., \bigOh{\K \log{\n}}.
\label{thm:singleton_doubleton_noisy}
\end{theorem}

\subsection{The Robustified Singleton-Only-SAFFRON}
The Singleton-Only-SAFFRON scheme can also be robustified in a similar manner. We design the signature matrix $U$ for the robustified SAFFRON scheme as follows.
\begin{align}
U &= 
\left[ \begin{array}{cccccc}
f(\bb_1) & f(\bb_2) & f(\bb_3) & \ldots & f(\bb_{n-2}) & f(\bb_{n-1}) \\ 
\rule{0pt}{3ex}   
f(\bb_{i_1}) & f(\bb_{i_2}) & f(\bb_{i_3}) & \ldots & f(\bb_{i_{\n-1}}) & f(\bb_{i_{\n}}) \\ 
\rule{0pt}{3ex}   
f(\bb_{j_1}) & f(\bb_{j_2}) & f(\bb_{j_3}) & \ldots & f(\bb_{j_{\n-1}}) & f(\bb_{j_{\n}})
\end{array}
\right]
\in \{0,1\}^{\frac{3\log_2 n}{R}\times n}
\end{align}
As explained before, by decoding the three sections of the measurement vector and checking whether three decoded indices satisfy \eqref{eq:permute_check}, the decoder can detect and resolve singletons. 

\begin{theorem}
With $m = 3e(1+\alpha)\beta(q)\nod \log{\nod} \log_2{\noi} \simeq 8.1548(1+\alpha)\beta(q) \nod \log{\nod} \log_2{\noi}$ tests, the robustified Singleton-Only-SAFFRON can find \emph{all} $\nod$ defective items with probability $1-$\bigOh{\frac{1}{K^\alpha}}, where $\beta(q) = \frac{1}{R} > \frac{1}{1-H(q)-\delta}$ for an arbitrarily small constant $\delta > 0$ and some constant $\alpha >0$. The computational complexity of the decoding algorithm is linear in the number of measurements, i.e., \bigOh{\nod \log{\nod} \log{\noi}}.
\label{thm:singleton_2_noisy}
\end{theorem}

\section{Simulation results} \label{sec:sim}
In this section, we evaluate the performance of the SAFFRON scheme and the robust SAFFRON scheme via extensive simulations. 
We implement simulators for both schemes in Python and test them on a laptop \footnote{We used a laptop with 2 GHz Intel Core i7 and 8 GB memory.}.

\subsection{SAFFRON}
The SAFFRON scheme recovers with high probability an arbitrarily-close-to-one fraction of $K$ defective items with \bigOh{\K \log \n} tests, as stated in Theorem \ref{thm:singleton_doubleton}.
The theorem also characterizes the optimal pairs of $(d^{\star}, \la^{\star})$ for a target recovery performance $\eps$: as the fraction of unidentified defective items $\eps$ decreases, the corresponding optimal left-degree $d^{\star}$ increases.
For different pairs of $(d, \M)$, we run SAFFRON $1000$ times and measure the average fraction of unidentified defective items: we choose $\n=2^{16}$, $\K = 100$, $d \in \{3,5,7,9\}$ and $\K \leq \M \leq 7\K$.
Plotted in Figure \ref{fig:num_of_bins} are the average fractions of unidentified defective items obtained via simulations for different values of $d$.
As expected, we can observe that if $\M$ is close to $\K$, the average fraction of unidentified defective items can be minimized by setting $d=3$, and if $\M$ is close to $7\K$, higher values of $d$ perform better. 
\begin{figure}
    \centering
   \includegraphics[width=0.5\textwidth]{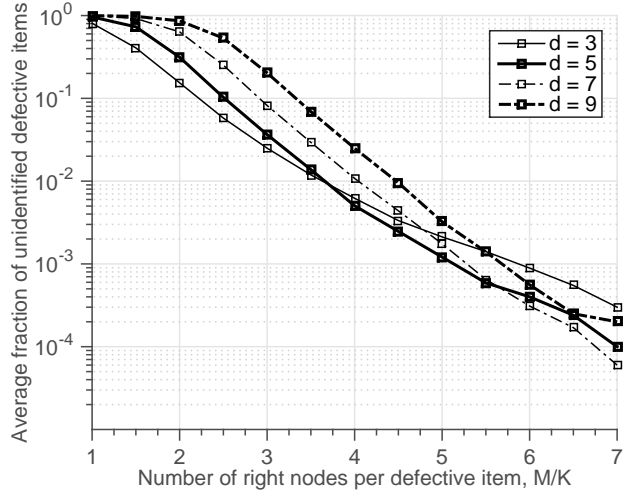}
    \caption{\textbf{The average fraction of unidentified defective items obtained via simulations.} For different pairs of $(d, \M)$, we simulated SAFFRON $1000$ times and measured the average fraction of unidentified defective items. We set $\n=2^{16}$, $\K = 100$, $d \in \{3,5,7,9\}$ and $\K \leq \M \leq 7\K$.}
    \label{fig:num_of_bins}
\end{figure}

We now simulate how computationally-efficient SAFFRON's decoding algorithm is.
We measure the average runtime of SAFFRON with $\n=2^{32}$, while increasing the value of $\K$. 
In Figure \ref{fig:time_complexity_a}, we plot the simulation results, and they clearly demonstrate the \bigOh{\K} factor of the computational complexity.
Similarly, we repeat simulations with $\K=2^{5}=32$, while increasing the value of $\n$. 
In Figure \ref{fig:time_complexity_b}, we plot the average runtime of SAFFRON with a logarithmic x-axis: we can clearly observe the \bigOh{\log\n} factor of the computational complexity.
\begin{figure}[h]
\centering
\begin{subfigure}{0.45\textwidth}
\includegraphics[width=\textwidth]{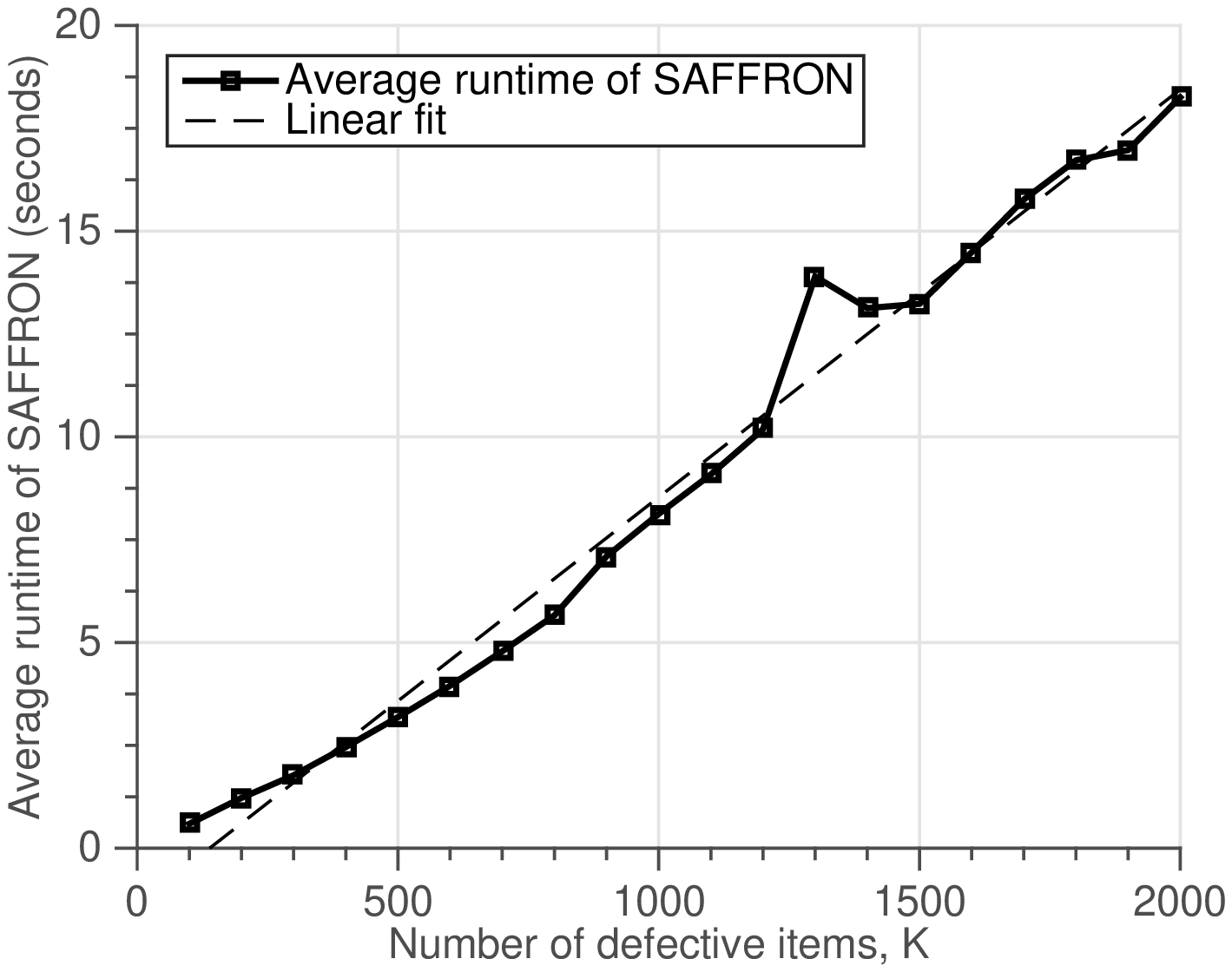}
\caption{Runtime with $\n = 2^{32}$ and varying $\K$.}
\label{fig:time_complexity_a}
\end{subfigure}
\begin{subfigure}{0.45\textwidth}
\includegraphics[width=\textwidth]{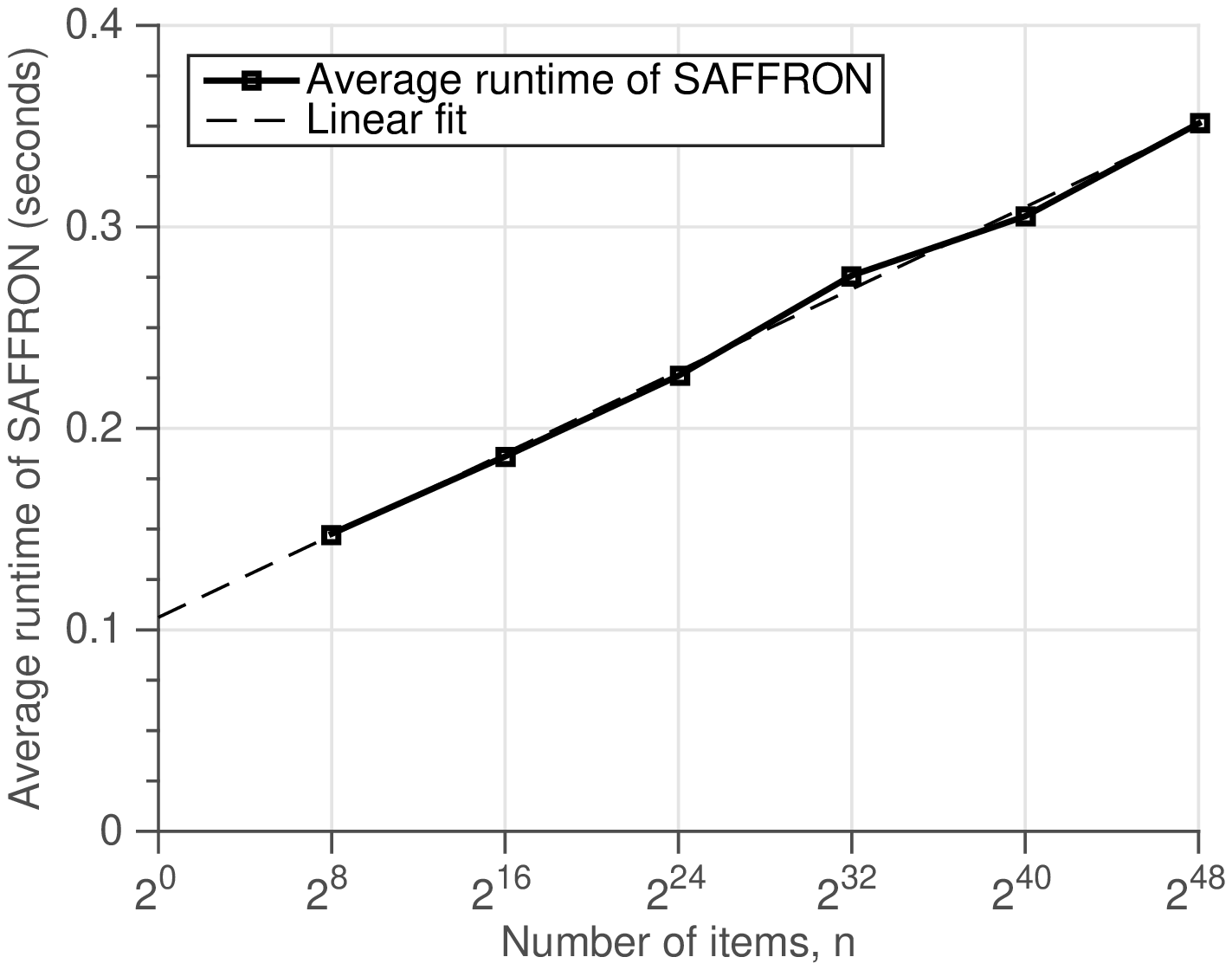}
\caption{Run-time with $\K = 2^{5}$ and varying $\n$.}
\label{fig:time_complexity_b}
\end{subfigure}
\caption{\textbf{Time complexity of SAFFRON.} We measure run-time of SAFFRON for varying values for $\n$ and $\K$.}
\label{fig:time_complexity}
\end{figure}

\subsection{Robustified-SAFFRON}
We now evaluate the robustified SAFFRON scheme. 
In our setting, we choose $\n = 2^{32} \simeq 4.3 \times 10^{9}$ and $\K = 2^{7} = 128$. For random bipartite graphs, we use $d = 12$ and $\M = 11.36\K$.
The noise model is the one we described in Section \ref{sec:noisy}. 
We vary the probability of error $q$ from $0.03$ to $0.05$: a test result is flipped with probability from $3\%$ to $5\%$. 

While we made use of capacity-achieving codes in Theorem \ref{thm:singleton_doubleton_noisy}, we use Reed-Solomon codes for simulations for simplicity \cite{costello2004error}.
A Reed-Solomon code takes a message of $c_k$ symbols from a finite field of size $c_q \geq c_n$, for a prime power $c_q$, and then encodes the message into $c_n$ symbols. 
This code can correct upto any $\lfloor \frac{c_n - c_k}{2} \rfloor$ symbol errors. 
By using a field of size $c_q = 2^8$, 
a binary representation of length $L~(= \log_2{\n})$ can be viewed as a $4$-symbol message, i.e., $c_k = 4$. 
Thus, the overall number of tests is as follows:
\begin{align}
m &= \underbrace{ 11.36\K}_\text{Number of right nodes} \times \underbrace{\frac{c_n}{c_k}}_\text{Error-correcting code expansion} \times \underbrace{6 \log_2 n}_\text{Number of message bits}. 
\end{align}
By having $c_n = c_k + 2t$, the robustified SAFFRON scheme can correct upto $t$ symbol errors within each section of the right-node measurement vector. Thus, we evaluate the performance of the robustified SAFFRON scheme with $c_n \in \{6,8,\ldots,16\}$ for various noise levels. 
We measure the average fraction of unidentified defective items over $1000$ runs for each setup. 
\begin{figure}[t]
    \centering    \includegraphics[width=0.6\textwidth]{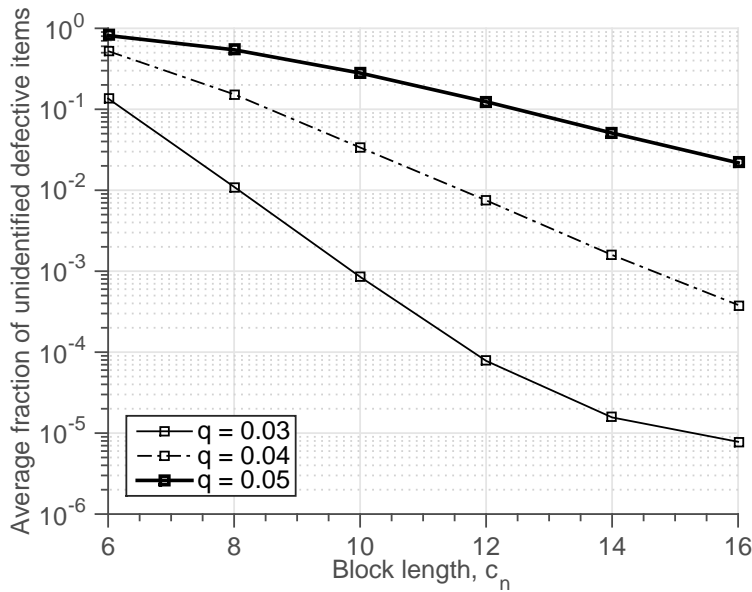}
    \caption{\textbf{Noisy simulation results.} We evaluate performance of our algorithm with noisy group testing results. For each pair of parameters, we measure the average fraction of missed defective items. 
    }
    \label{fig:noisy}
\end{figure}
Figure \ref{fig:noisy} shows the simulation results; the $x$-axis is the block-length of the used code, and the logarithmic $y$-axis is the average fraction of unidentified defective items. 
We can observe that for a higher noise level $q$, the minimum block length required to achieve a certain value of $\eps$ increases. 

We also observe that \emph{the robustified SAFFRON perfectly recovers all defective items even with the presence of erroneous test results} with certain parameters. 
For instance, with $q = 0.02$, we test the robustified SAFFRON $1000$ times with $c_n = 12$. 
For all the test cases, it successfully recovers all $\K=128$ defective items from the population of $\n \simeq 4.3\times 10^9$ items
with $\m = 838080 \simeq \frac{2}{10000}\n$ tests. 
Further, the decoding time takes only about $3.8$ seconds on average. 
Similarly, we observe the perfect recovery with $q = 0.01$ and $c_n \geq 12$ and with $q = 0.005$ and $c_n \geq 6$.

\section{Conclusion}
In this paper, we have proposed SAFFRON (\textbf{S}parse-gr\textbf{A}ph codes \textbf{F}ramework \textbf{F}or g\textbf{RO}up testi\textbf{N}g), which recovers an arbitrarily-close-to-one $(1-\eps)$-fraction of $\K$ defective items with high probability with $6C(\eps)K\log_2{\n}$ tests, where $C(\eps)$ is a relatively small constant that depends only on $\eps$. 
Also, the computational complexity of the decoding algorithm of SAFFRON is order-optimal. 
We have described the design and analysis of SAFFRON based on the powerful modern coding-theoretic tools of sparse-graph coding and density evolution. 
We have also proposed a variant of SAFFRON, Singleton-Only-SAFFRON, which recovers all defective items with $2e(1+\alpha)\K \log\K \log_2 \n$ tests, with probability $1-$\bigOh{\frac{1}{K^\alpha}}. 
Further, we robustify SAFFRON and Singleton-Only-SAFFRON by using modern error-correcting codes so that they can recover the set of defective items with noisy test results. 
To support our theoretical results, we have provided extensive simulation results that validate the theoretical efficacy and the practical potential of SAFFRON.

\bibliographystyle{IEEEtran}
\bibliography{ref}
\end{document}